\documentclass[12pt]{article}
\usepackage[margin=1in]{geometry}

\usepackage{graphicx}
\usepackage[unicode]{hyperref}
\usepackage{amsthm}
\usepackage{amssymb}
\usepackage{amsmath}
\usepackage{mathrsfs}
\usepackage{verbatim}
\usepackage{algorithm}
\usepackage{mathtools}
\usepackage{xspace}
\usepackage{enumitem}
\usepackage{color}
\usepackage{url}
\usepackage[all]{xy}
\usepackage{xspace}
\usepackage[all]{xy}
\usepackage{amsmath}
\usepackage{mathtools}
\usepackage{verbatim}

\newcommand{\rid}{\overline{\ri}}

\newcommand{\distance}[2]{\|#1-#2\|}
\newcommand{\ri}{\mathcal{R}}
\newcommand{\eps}{\varepsilon}
\newcommand{\dd}{\Delta}
\newcommand{\Cech}{\v{C}ech\xspace}
\newcommand{\cech}{\mathcal{C}}

\newcommand{\R}{\mathbb{R}}
\newcommand{\Z}{\mathbb{Z}}
\newcommand{\dotp}[2]{\vec{#1}\cdot\vec{#2}}
\newcommand{\Del}{\mathcal{D}}
\newcommand{\ignore}[1]{}

\newtheorem{lemma}{Lemma}

\newtheorem{theorem}[lemma]{Theorem}
\newtheorem{proposition}[lemma]{Proposition}
\newtheorem{corollary}[lemma]{Corollary}

\begin{document}

\title{Polynomial-Sized Topological Approximations\\Using The Permutahedron}
\author{
Aruni Choudhary\footnote{Max Planck Institute for Informatics, 
Saarbr\"ucken, Germany \texttt{(aruni.choudhary@mpi-inf.mpg.de)}}
\and 
Michael Kerber\footnote{Graz University of Technology, 
Graz, Austria\texttt{(kerber@tugraz.at)}} 
\and
Sharath Raghvendra\footnote{Virginia Tech, 
Blacksburg, USA \texttt{(sharathr@vt.edu)}}}
\date{}
\maketitle

\begin{abstract}
Classical methods to model topological properties of point clouds,
such as the Vietoris-Rips complex, suffer from the combinatorial explosion
of complex sizes.
We propose a novel technique to approximate a multi-scale filtration
of the Rips complex with improved bounds for size: precisely,
for $n$ points in $\mathbb{R}^d$, we obtain a $O(d)$-approximation with at most 
$n2^{O(d \log k)}$ simplices of dimension $k$ or lower. 
In conjunction with dimension reduction techniques,
our approach yields a $O(\mathrm{polylog} (n))$-approximation 
of size $n^{O(1)}$ for Rips filtrations on arbitrary metric spaces. 
This result stems from high-dimensional lattice
geometry and exploits properties of the permutahedral lattice, a well-studied
structure in discrete geometry. 

Building on the same geometric concept,
we also present a lower bound result
on the size of an approximate filtration:
we construct a point set 
for which every $(1+\eps)$-approximation of the \Cech filtration
has to contain $n^{\Omega(\log\log n)}$ features,
provided that $\eps <\frac{1}{\log^{1+c} n}$ for $c\in(0,1)$.
\end{abstract}

\section{Introduction}
\label{section:intro}

\subparagraph*{Motivation and previous work}
Topological data analysis aims at finding and reasoning about the 
underlying topological features of metric spaces. 
The idea is to represent a data set by a set of discrete structures on a range of scales and 
to track the evolution of homological features
as the scale varies. The theory of \emph{persistent} homology allows for a topological summary,
called the \emph{persistence diagram} which summarizes the lifetimes of topological
features in the data as the scale under consideration varies monotonously.
A major step in the computation of this topological signature is the question of how to compute
a \emph{filtration}, that is, a multi-scale representation of a given data set.

For data in the form of finite point clouds,
two frequently used constructions are the \emph{(Vietoris-)Rips} complex $\ri_\alpha$ 
and the \emph{\Cech} complex $\cech_\alpha$
which are defined with respect to a scale parameter $\alpha\geq 0$.
Both are simplicial complexes capturing 
the proximity of points at scale $\alpha$, with different levels of accuracy.
Increasing $\alpha$ from $0$ to $\infty$ yields a nested sequence of simplicial complexes called a \emph{filtration}. 

Unfortunately, Rips and \v{C}ech filtrations 
can be uncomfortably large to handle.
For homological features in low dimensions, it suffices to consider
the $k$-skeleton of the complex, that is, 
all simplices of dimension at most $k$.
Still, the $k$-skeleton of Rips and \Cech complexes 
can be as large as $n^{k+1}$ for $n$ points,
which is already impractical for small $k$ when $n$ is large.
One remedy is to construct an \emph{approximate filtration}, 
that is, a filtration that yields a similar topological signature as the original filtration
but is significantly smaller in size. The notion of ``similarity'' in this context can be made formal through
a distance measure on persistence diagrams. The most frequently used similarity measure 
is the \emph{bottleneck distance},
which finds correspondences between topological features of two filtrations, such that the
lifetimes of each pair of matched features are as close as possible.
A related notion is the \emph{log-scale bottleneck distance} which allows a larger discrepancy
for larger scales and thus can be seen as a relative approximation,
with usual bottleneck distance as its absolute counterpart. We call an approximate filtration
a \emph{$c$-approximation} of the original, if their persistence diagrams have 
log-scale bottleneck distance at most $c$.

Sheehy~\cite{sheehy-rips} gave the first such approximate filtration for Rips complexes with a formal guarantee.
For $0<\eps\le 1/3$, he constructs a $(1+\eps)$-approximate filtration of the Rips filtration.
The size of its $k$-skeleton is only $n(\frac{1}{\eps})^{O(\dd k)}$, where $\dd$ is the doubling dimension of the metric.
Since then, several alternative technique have been explored for Rips~\cite{desu-gic} and \Cech complexes~\cite{bs-approximating,cks-approximate,kbs-cech}, all arriving at the same complexity bound.

While the above approaches work well for instances 
where $\dd$ and $k$ are small, 
we focus on high-dimensional point sets. 
This has two reasons: first, one might simply want to analyze data sets
for which the intrinsic dimension is high, but the existing methods do
not succeed in reducing the complex size sufficiently. Second, even
for medium-size dimensions, one might not want to restrict its scope
to the low-homology features, so that $k=\dd$ is not an unreasonable
parameter choice.
To adapt the aforementioned schemes to play nice with high dimensional point clouds,
it makes sense to use dimension reduction results to
eliminate the dependence on $\dd$. Indeed, it has been shown, in analogy to the famous
Johnson-Lindenstrauss Lemma~\cite{jl-lemma}, that an orthogonal projection to a
$O(\log n/\eps^2)$-dimensional subspace yields another $(1+\eps)$ approximate filtration
~\cite{kr-approximation,sheehy-persistent}.
Combining these two approximation schemes, however, yields an approximation of size
$O(n^{k+1})$ (ignoring $\eps$-factors) and does not improve upon the exact case.

\subparagraph{Our contributions}
We present two results about the approximation of Rips and \Cech filtrations:
we give a scheme for approximating the Rips filtration with smaller complex
size than existing approaches, at the price of guaranteeing only an 
approximation quality of $\mathrm{polylog}(n)$.
Since Rips and \Cech filtrations 
approximate each other by a constant factor,
our result also extends to the \Cech filtration, with an additional constant factor
in the approximation quality.
Second, we prove that any approximation scheme for the \Cech filtration has
superpolynomial size in $n$ if high
accuracy is required. For this result, our proof technique does not extend
to Rips complexes.
In more detail, our results are as follows:

\textbf{Upper bound}:
We present a $6(d+1)$-approximation of the Rips filtration 
for $n$ points in $\R^d$
whose $k$-skeleton has a size of $n2^{O(d\log k)}$ on each scale.
This shows that by using a more rough approximation,
we can achieve asymptotic improvements on the complex size.
The real power of our approach reveals itself in high dimensions,
in combination with dimension reduction techniques.
In conjunction with the lemma of Johnson and Lindenstrauss~\cite{jl-lemma},
we obtain an $O(\log n)$-approximation 
with size $n^{O(\log k)}$ at any scale,
which is much smaller than the original filtration; however, for the complete case $k=\log n$, the bound is still super-polynomial in $n$.
Combined with a different dimension reduction result of Matou\v{s}ek~\cite{mt-embedding}, we obtain a 
$O(\log^{3/2} n )$-approximation of size $n^{O(1)}$.
This is the first polynomial bound in $n$ of an approximate filtration,
independent of the dimensionality of the point set.
For inputs from arbitrary metric spaces (instead of points in $\R^d$), the same results hold with an additional $O(\log n)$ factor in the
approximation quality.

Our approximations are discrete, and the number of scales that have to be
considered is determined by the logarithm of the spread of the point set
(the ratio of diameter and closest point distance). In this work, we tacitly
assume the spread to be constant, and concentrate on the complex size
on a fixed scale as our quality measurement.

\textbf{Lower bound}:
We construct a point set of $n$ points in $d=\Theta(\log n)$ dimensions
whose \Cech filtration has $n^{\Omega(\log\log n)}$
persistent features with ``relatively long'' lifetime.
Precisely, that means that any $(1+\delta)$-approximation 
has to contain a bar of non-zero length for each of those features
if $\delta<O(\frac{1}{\log^{1+c}n})$ with $c\in (0,1)$.
This shows that it is impossible to define an approximation scheme
that yields an accurate approximation of the \Cech complexes as well
as polynomial size in $n$.

\textbf{Methods}:
Our results follow from a link to lattice geometry: the $A^\ast$-lattice is a configuration of points in $\R^d$
which realizes the thinnest known coverings for low dimensions~\cite{sloway-book}.
The dual Voronoi polytope of
a lattice point is the \emph{permutahedron}, whose vertices are obtained by all coordinate permutations of a fixed
point in $\R^d$. 

Our technique resembles the perhaps simplest approximation scheme
for point sets: if we digitize $\R^d$ with $d$-dimensional pixels,
we can take the union of pixels that contain input points as our
approximation. Our approach does the same, except that 
we use a tessellation of permutahedra for digitization.
In $\R^2$, our approach corresponds to the common approach of
replacing the square tiling by a hexagonal tiling. 
We exploit that the permutahedral tessellation is in generic position,
that is, no more than $d+1$ polytopes have a common intersection.
At the same time, permutahedra are still relatively round, that is, they 
have small diameter and non-adjacent polytopes are well-separated.
These properties ensure good approximation quality and a small complex.
In comparison, a cubical tessellation yields a $O(\sqrt{d})$-approximate
Rips filtration which looks like an improvement over our $O(d)$-approximation,
but the highly degenerate configuration of the cubes yields
a complex size of $n2^{O(dk)}$, and therefore does not constitute
an improvement over Sheehy's approach~\cite{sheehy-rips}.

For the lower bound, we arrange $n$ points in a way that one center
point has the permutahedron as Voronoi polytope, and we consider simplices
incident to that center point in a fixed dimension. We show
a superpolynomial number of these simplices 
create or destroy topological features 
of non-negligible persistence. 

\subparagraph*{Outline of the paper}
We begin by reviewing basics of persistent homology in Section~\ref{section:backgnd}. Next, we study several relevant properties
of the $A^*$ lattice in Section~\ref{section:permuto}. An approximation algorithm based on concepts from
Section~\ref{section:permuto} is presented in Section~\ref{section:appsch}. In Section~\ref{section:lowerbnd}, we present
the lower bound result on the size of \Cech filtrations. We conclude in Section~\ref{section:concl}.

\section{Topological background}
\label{section:backgnd}

We review some topological concepts needed in our argument.
More extensive treatments covering most of the material
can be found in the textbooks~\cite{brunrer-book,hatcher,munkres}.

\subparagraph{Simplicial complexes}
For an arbitrary set $V$, called \emph{vertices}, a \emph{simplicial complex}
over $V$ is a collection of non-empty subsets which is closed 
under taking non-empty subsets. The elements of a simplicial complex $K$
are called \emph{simplices} of $K$. A simplex $\sigma$ is a \emph{face}
of $\tau$ if $\sigma\subseteq \tau$. A \emph{facet} is a face of co-dimension $1$.
The dimension of $\sigma$ is $k:=|\sigma|-1$;
we also call $\sigma$ a $k$-simplex in this case.
The \emph{$k$-skeleton} of $K$
is the collection of all simplices of dimension at most $k$.
For instance, the $1$-skeleton of $K$ is a graph 
defined by its $0$- and $1$-simplices.

We discuss two ways of generating simplicial complexes. In the first one,
take a collection $\mathcal{S}$ 
of sets over a common universe (for instance, polytopes
in $\R^d$), and define the \emph{nerve} of $\mathcal{S}$ as the
simplicial complex whose vertex set is $\mathcal{S}$, and a $k$-simplex
$\sigma$ is in the nerve if the corresponding $(k+1)$ sets have a non-empty
common intersection. The \emph{nerve theorem}~\cite{borsuk-nerve}
states that if all sets in $\mathcal{S}$ are convex subsets of $\R^d$,
their nerve is homotopically equivalent to the union of the sets
(the statement can be generalized significantly; see~\cite[Sec. 4.G]{hatcher}).
The second construction that we consider are \emph{flag complexes}:
Given a graph $G=(V,E)$, we define a simplicial complex $K_G$ over the vertex
set $V$ such that a $k$-simplex $\sigma$ is in $K$ if for every distinct pair of
vertices $v_1, v_2\in \sigma$, the edge $(v_1,v_2)$ is in $E$. 
In other words, $K_G$ is the maximal simplicial complex with $G$ 
as its $1$-skeleton. In general, a complex $K$ is called a flag complex, if $K=K_G$
with $G$ being the $1$-skeleton of $K$.

Given a set of points $P$ in $\R^d$ and a parameter $r$, the 
\emph{\Cech complex at scale $r$}, $\cech_r$
is defined as the nerve of the balls centered 
at the elements of $P$, each of radius $r$.
This is a collection of convex sets. Therefore, the nerve theorem is applicable and
it asserts that the nerve agrees homotopically with the union of balls.
In the same setup, we can as well consider the intersection graph $G$ of
the balls (that is, we have an edge between two points
if their distance is at most $2r$). The flag complex of $G$ is called
the \emph{(Vietoris-)Rips complex at scale $r$}, denoted by $\ri_r$.
The relation
$\cech_r\subseteq\ri_r\subseteq\cech_{\sqrt{2}r}$
follows from Jung's Theorem~\cite{jung}.

\subparagraph*{Persistence Modules and simplicial filtrations}
A \emph{persistence module} $(V_\alpha)_{\alpha\in G}$ 
for a totally ordered index set $G\subseteq\R$ is a sequence of vector spaces
with linear maps $F_{\alpha,\alpha'}:V_\alpha\rightarrow V_{\alpha'}$
for any $\alpha\leq\alpha'$, satisfying
$F_{\alpha,\alpha}=id$
and $F_{\alpha',\alpha''}\circ F_{\alpha,\alpha'}=F_{\alpha,\alpha''}$.
Persistence modules can be decomposed into \emph{indecomposable intervals}
giving rise to a \emph{persistent barcode} which is a complete discrete
invariant of the corresponding module.

A distance measure between persistence modules is defined through 
interleavings: we call two modules $(V_\alpha)$ and $(W_\alpha)$
with linear maps $F_{\cdot,\cdot}$ and $G_{\cdot,\cdot}$
\emph{additively $\eps$-interleaved}, if there exist linear maps
$\phi:V_\alpha\rightarrow W_{\alpha+\eps}$ and $\psi:W_\alpha\rightarrow
V_{\alpha+\eps}$ such that the maps $\phi$ and $\psi$ commute with $F$
and $G$~(see~\cite{ch-proximity}).
We call the modules \emph{multiplicatively $c$-interleaved}
with $c\geq 1$, if there exist linear maps
$\phi:V_\alpha\rightarrow W_{c\alpha}$ and $\psi:W_\alpha\rightarrow
V_{c\alpha}$ with the same commuting properties.
Equivalently, this means that the modules are additively $(\log c)$-interleaved
when switching to a logarithmic scale.
In this case, we also call the module $(G_\alpha)$
a \emph{$c$-approximation} of $(F_\alpha)$ (and vice versa).
Note that the case $c=1$ implies that the two modules give rise to 
the same persistent barcode, which is usually referred to as 
the \emph{persistence equivalence theorem}~\cite{brunrer-book}.

The most common way to generate persistence modules is through
the homology of sequences of simplicial complexes:
a \emph{(simplicial) filtration} $(K_\alpha)_{\alpha\in G}$ 
over a totally order index set $G\subseteq\R$
is a sequence of simplicial complexes connected by simplicial maps 
$f_{\alpha,\alpha'}:K_{\alpha}\rightarrow K_{\alpha'}$ 
for any $\alpha\leq \alpha'$, such that $f_{\alpha,\alpha}=id$
and $f_{\alpha',\alpha''}\circ f_{\alpha,\alpha'}=f_{\alpha,\alpha''}$.
By the functorial properties of homology (using some fixed field $\mathbb{F}$
and some fixed dimension $p\geq 0$),
such a filtration gives rise to a persistence module 
$(H_p(K_\alpha,\mathbb{F}))_{\alpha\in G}$.
We call a filtration a $c$-approximation of another filtration
if the corresponding persistence modules induced by homology 
are $c$-approximations of each other. 

The standard way of obtaining a filtration
is through a nested sequence of simplicial complexes, where the simplicial
maps are induced by inclusion. Examples are the \emph{\Cech filtration}
$(\cech_\alpha)_{\alpha\in\R}$ and the \emph{Rips filtration}
$(\ri_\alpha)_{\alpha\in\R}$. By the relations of Rips and \Cech complexes
from above, the Rips filtration is a $\sqrt{2}$-approximation of the \Cech
filtration.

\subparagraph*{Simplex-wise \Cech filtrations and (co-)face distances}
In the \Cech filtration $(\cech_\alpha)$, every simplex
has an \emph{alpha value} $\alpha_\sigma:=\min\{\alpha\geq 0\mid \sigma\in \cech_\alpha\}$, which equals the radius of the minimal enclosing ball of
its boundary vertices.
If the point set $P$ is finite, the \Cech filtration consists 
of a finite number of simplices, and we can define 
a \emph{simplex-wise filtration}
\[\emptyset=\cech^0\subsetneq \cech^2\subsetneq\ldots\subsetneq \cech^m, \]
where exactly one simplex is added  from $\cech^i$ to $\cech^{i+1}$,
and where $\sigma$ is added before $\tau$ whenever $\alpha_\sigma<\alpha_\tau$.
The filtration is not unique and ties can be broken arbitrarily.

In a simplex-wise filtration, passing from $\cech^i$ to $\cech^{i+1}$ means adding
the $k$-simplex $\sigma:=\sigma_{i+1}$.
The effect of this addition is that either
a $k$-homology class comes into existence, or a $(k-1)$-homology
class is destroyed. Depending on the case, we call $\sigma$
\emph{positive} or \emph{negative}, accordingly.
In terms of the corresponding persistent barcode,
there is exactly one interval \emph{associated to $\sigma$}
either starting at $i$ (if $\sigma$ is positive)
or ending at $i$ (if $\sigma$ is negative).
We define the \emph{(co-)face distance} $L_\sigma$ ($L^\ast_\sigma$) 
of $\sigma$ as the minimal distance between $\alpha_\sigma$ and its (co-)facets,
\[L_\sigma:=\min_{\tau \text{ facet of }\sigma} \alpha_\sigma-\alpha_\tau\quad\quad L^\ast_\sigma:=\min_{\tau \text{ co-facet of }\sigma} \alpha_\tau-\alpha_\sigma.\]
Note that $L_\sigma$ and $L_\sigma^\ast$ can be zero.
Nevertheless, they constitute lower bounds for the persistence
of the associated barcode intervals. 
An alternative to our proof is 
to argue using structural properties
of the matrix reduction algorithm for persistent homology~\cite{brunrer-book}.
\begin{lemma}
\label{lem:L_bound}
If $\sigma$ is negative, the barcode interval associated to $\sigma$
has persistence at least $L_\sigma$.
\end{lemma}
\begin{proof}
$\sigma$ kills a $(k-1)$-homology class by assumption,
and this class is represented by the cycle $\partial\sigma$.
However, this cycle came into existence when the last facet $\tau$ of $\sigma$
was added. Therefore, the lifetime of the cycle destroyed by $\sigma$
is at least $\alpha_\sigma-\alpha_\tau$.
\end{proof}

\begin{lemma}
\label{lem:L_star_bound}
If $\sigma$ is positive, the homology class created by $\sigma$
has persistence at least $L^\ast_\tau$
\end{lemma}
\begin{proof}
$\sigma$ creates a $k$-homology class; every representative cycle of this
class is non-zero for $\sigma$. To turn such a cycle into a boundary,
we have to add a $(k+1)$-simplex $\tau$ with $\sigma$ in its boundary
(otherwise, any $(k+1)$-chain formed will be zero for $\sigma$). Therefore,
the cycle created at $\sigma$ persists 
for at least $\alpha_\tau-\alpha_\sigma$.
\end{proof}

\section{The \texorpdfstring{$A^\ast$}{}-lattice and the permutahedron}
\label{section:permuto}

A \emph{lattice} $L$ in $\R^d$ is the set of all integer-valued linear combination of $d$ independent vectors,
called the \emph{basis} of the lattice. Note that the origin belongs to every lattice.
The \emph{Voronoi polytope} of a lattice $L$ is the closed set of all points in $\R^d$ for which 
the origin is among the closest lattice points. Since lattices are invariant under translations, the Voronoi polytopes
for other lattice points are just translations of the one at the origin, and these polytopes tile $\R^d$.
An elementary example is the integer lattice, spanned by the unit vectors $(e_1,\ldots,e_d)$,
whose Voronoi polytope is the unit $d$-cube, shifted by $(-1/2)$ in each coordinate direction.

We are interested in a different lattice, called the $A_d^\ast$-lattice, whose properties are also well-studied~\cite{sloway-book}.
First, we define the $A_d$ lattice as the set of points $(x_1,\cdots,x_{d+1})\in \Z^{d+1}$ satisfying $\sum_{i=1}^{d+1} x_i=0$.
$A_d$ is spanned by vectors of the form $(e_i,-1)$, $i=1,\ldots,d$.
While it is defined in $\R^{d+1}$, all points lie on the hyperplane $H$
defined by $\sum_{i=1}^{d+1} y_i = 0$.
After a suitable change of basis, we can express $A_d$ by $d$ vectors in $\R^d$; thus, it is indeed a lattice.
In low dimensions, $A_2$ is the hexagonal lattice, and $A_3$ is the FCC lattice that realizes the best sphere packing
configuration in $\R^3$~\cite{hales-kepler}.

The \emph{dual lattice} $L^\ast$ of a lattice $L$ is defined 
as the set of points $(y_1,\ldots,y_{d})$ in $\R^{d}$ such that
$y\cdot x\in\Z$ for all $x\in L$~\cite{sloway-book}. Both the integer
lattice and the hexagonal lattice are self-dual, 
while the dual of $A_3$ is the BCC lattice 
that realizes the thinnest sphere covering configuration among lattices in $\R^3$~\cite{bambah}.

We are mostly interested in the Voronoi polytope $\Pi_d$ generated by $A^\ast_d$.
Again, the definition becomes easier when embedding $\R^d$ one dimension higher as the hyperplane $H$.
In that representation, it is known~\cite{sloway-book} that 
$\Pi_d$ has $(d+1)!$ vertices obtained
by all permutations of the coordinates of
\[\frac{1}{2(d+1)}(d,d-2,d-4,\cdots,-d+2,-d).\]
$\Pi_d$ is known 
as the \emph{permutahedron}~\cite[Lect. 0]{zg-polytopes}.%
\footnote{Often, a scaled, translated  and rotated version is considered,
in which all permutations of the point $(1,\ldots,d+1)$ are taken.}
Our approximation results in Section~\ref{section:appsch} and~\ref{section:lowerbnd}
are based on various combinatorial and geometric properties of $\Pi_d$, which we describe next.
We will fix $d$ and write $A^\ast:=A^\ast_d$ and $\Pi:=\Pi_d$ for brevity.

\subparagraph*{Combinatorics}
The $k$-faces of $\Pi$ correspond to ordered partitions of the coordinate indices $[d+1]:=\{1,\cdots,d+1\}$
into $(d+1-k)$ non-empty subsets $S_1,\cdots,S_{d+1-k}$ such that all coordinates in $S_i$ are smaller
than all coordinates in $S_j$ for $i<j$~\cite{zg-polytopes}.
For example, with $d=3$, the partition $(\{1,3\},\{2,4\})$ is the $2$-face spanned by all points
for which the two smallest coordinates appear at the first and the third position.
This is an example of a facet of $\Pi$, for which we need to partition the indices in exactly $2$ subsets;
equivalently, the facets of $\Pi$ are in one-to-one correspondence to non-empty proper subsets of $[d+1]$
so $\Pi$ has $2^{d+1}-2$ facets.
The vertices of $\Pi$ are the $(d+1)$-fold ordered partitions which correspond to permutations of $[d+1]$,
reassuring the fact that $\Pi$ has $(d+1)!$ vertices.
Moreover, two faces $\sigma$, $\tau$ of $\Pi$ with $\dim\sigma < \dim\tau$ are incident if the partition of $\sigma$
is a refinement of the partition of $\tau$. Continuing our example from before, the four $1$-faces bounding
the $2$-face $(\{1,3\},\{2,4\})$ are $(\{1\},\{3\},\{2,4\})$,$(\{3\},\{1\},\{2,4\})$, $(\{1,3\},\{2\},\{4\})$,
and  $(\{1,3\},\{4\},\{2\})$. 
Vice versa, we obtain co-faces of a face by combining consecutive partitions
into one larger partition. 
For instance, the two co-facets of $(\{1,3\},\{4\},\{2\})$
are $(\{1,3\},\{2,4\})$ and $(\{1,3,4\},\{2\})$.
\begin{lemma}
\label{lem:pi_adjacency}
Let $\sigma$ and $\tau$ be two facets of $\Pi$, defined by the partitions $(S_\sigma,[d+1]\setminus S_\sigma)$
and $(S_\tau,[d+1]\setminus S_\tau)$, respectively. 
Then $\sigma$ and $\tau$ are adjacent in $\Pi$ iff $S_\sigma\subseteq S_\tau$
or $S_\tau\subseteq S_\sigma$.
\end{lemma}
\begin{proof}
Two facets are adjacent if they share a common face. By the properties of the permutahedron,
this means that the two facets are adjacent if and only if their partitions permit a common refinement,
which is only possible if one set is contained in the other.
\end{proof}
We have already established that $\Pi$ has ``few'' ($2^{d+1}-2=O(2^d)$)
$(d-1)$-faces and ``many'' ($(d+1)!=O(2^{d\log d})$) $0$-faces. 
We give an interpolating bound for all intermediate dimensions.

\begin{lemma}
\label{lem:perm_no_of_faces}
The number of $(d-k)$-faces of $\Pi$ is bounded by $2^{2 (d+1)\log_2 (k+1)}$.
\end{lemma}
\begin{proof}
By our characterization of faces of $\Pi$, it suffices to count
the number of ordered partitions of $[d+1]$ into $(k+1)$ subsets.
That number equals $(k+1)!$ times the number of unordered partitions.
The number of unordered partitions, in turn, is known as 
\emph{Stirling number of the second kind}~\cite{rd-stirling}
and bounded by $\frac{1}{2}\binom{d+1}{k+1}(k+1)^{d-k}$.
Multiplying with $(k+1)!$ yields an upper bound for $(d-k)$-faces,
which can be bounded by $(k+1)^{2(d+1)}$ for $k\le d$. 
\end{proof}

\subparagraph*{Geometry}
All vertices of $\Pi$ are equidistant from the origin,
and it can be checked with a simple calculation that this distance is $\sqrt{\frac{d(d+2)}{12(d+1)}}$.
Using the triangle inequality, we obtain:
\begin{lemma}
\label{lemma:pmdiam}
 The diameter of $\Pi$ is at most $\sqrt{d}$.
\end{lemma}
The permutahedra centered at all lattice points of $A^\ast$ define the Voronoi tessellation of $A^\ast$.
Its nerve is the Delaunay triangulation $\Del$ of $A^\ast$. An important property of $A^\ast$ is that,
unlike for the integer lattice, 
$\Del$ is non-degenerate~-- this will ultimately ensure small upper bounds
for the size of our approximation scheme.

\begin{lemma}
\label{lemma:pmgenp}
 Each vertex of a permutahedral cell has precisely $d+1$ cells adjacent to it. 
 In other words, the $A^*_d$ lattice points are in general position.
\end{lemma}
The proof idea is to look at any vertex of the Voronoi cell and argue that it has precisely $(d+1)$
equidistant lattice points. See~\cite[Thm.2.5]{stanford-tech} for a concise, or the appendix for a detailed argument.
As a consequence, we can identify Delaunay simplices incident to the origin 
with faces of $\Pi$.

\begin{proposition}
\label{prop:equiv}
The $(k-1)$-simplices in $\Del$ that are incident to the origin are in one-to-one-correspondence to the
$(d-k+1)$-faces of $\Pi$ and, hence, in one-to-one correspondence
to the ordered $k$-partitions of $[d+1]$.
\end{proposition}
Let $V$ denote the set of lattice points that share a Delaunay edge
with the origin. 
The following statement shows that the point set $V$ is in convex position,
and the convex hull encloses $\Pi$ with some ``safety margin''. 
The proof is a mere calculation, deriving an explicit equation
for each hyperplane supporting the convex hull and applying it
to all vertices of $V$ and of $\Pi$.
The argument is detailed in the appendix.
\begin{lemma}
 \label{lem:hplane_facet}
 For each $d$-simplex attached to the origin, the facet $\tau$
 opposite to the origin
 lies on a hyperplane which is at least a distance $\frac{1}{\sqrt{2}(d+1)}$
 to $\Pi$ and all points of $V$ are either on the hyperplane or
 on the same side as the origin.
\end{lemma}

\begin{lemma}
\label{lem:pmdist}
If two lattice points are not adjacent in $\Del$, the corresponding Voronoi polytopes have a distance of at least
$\frac{\sqrt{2}}{d+1}$.
\end{lemma}
\begin{proof}
Lemma~\ref{lem:hplane_facet} shows that $\Pi$ is contained in a convex
polytope $C$ and the distance of $\Pi$ to the boundary of $C$
is at least $\frac{1}{\sqrt{2}(d+1)}$. Moreover, if $\Pi'$ is
the Voronoi polytope of a non-adjacent lattice point $o'$, the corresponding
polytope $C'$ is interior-disjoint from $C$. To see that, note that
the simplices in $\Del$ incident to the origin triangulate the interior
of $C$, and likewise for $o'$ any interior intersection would be covered
by a simplex incident to $o$ and one incident to $o'$, and since they are 
not connected, the simplices are distinct, contradicting the fact that
$\Del$ is a triangulation. Having established that $C$ and $C'$
are interior-disjoint, the distance between $\Pi$ and $\Pi'$
is at least $\frac{2}{\sqrt{2}(d+1)}$, as required.
\end{proof}

Recall the definition of a flag complex as the maximal simplicial complex one can form from a given graph.
We next show that $\Del$ is of this form. While our proof exploits
certain properties of $A^\ast$, we could not exclude the possibility
that the Delaunay triangulation of any lattice is a flag complex.
\begin{lemma}
 \label{lem:flag_prop}
 $\Del$ is a flag complex. 
\end{lemma}
\begin{proof}
The proof is based on two claims: consider two facets $f_1$ and $f_2$ of $\Pi$
that are disjoint, that is, do not share a vertex. 
In the tessellation, there are permutahedra $\Pi_1$ attached to $f_1$
 and $\Pi_2$ attached to $f_2$.
The first claim is that $\Pi_1$ and $\Pi_2$ are disjoint.
We prove this explicitly by constructing a hyperplane separating $\Pi_1$
and $\Pi_2$. See the appendix for further details.

The second claim is that if $k$ facets of $\Pi$ are pairwise intersecting,
they also have a common intersection. Another way to phrase this statement
is that the link of any vertex in $\Del$ is a flag complex.
This is a direct consequence of Lemma~\ref{lem:pi_adjacency}. 
See the appendix for more details.

The lemma follows directly with these two claims: consider $k+1$
vertices of $\Del$ which pairwise intersect. 
We can assume that one point is the origin, and the other $k$ points
are the centers of permutahedra that intersect $\Pi$ in a facet.
By the contrapositive of the first claim, all these facets have to intersect
pairwisely, because all vertices have pairwise Delaunay edges.
By the second claim, there is some common vertex of $\Pi$ to all these facets,
and the dual Delaunay simplex contains the $k$-simplex spanned by the vertices.
\end{proof}

\section{Approximation scheme}
\label{section:appsch}
Given a point set $P$ of $n$ points in $\R^d$, 
we describe our approximation complex $X_\beta$ for a fixed scale $\beta>0$.
For that, let $L_\beta$ denote the $A_d^*$ lattice in $\R^d$, with each lattice vector
scaled by $\beta$.
Recall that the Voronoi cells of the lattice points are scaled permutahedra which tile $\R^d$.
The bounds for the diameter (Lemma~\ref{lemma:pmdiam})
as well as for the distance between non-intersecting Voronoi polytopes (Lemma~\ref{lem:pmdist})
remain valid when multiplying them with the scale factor.
Hence, any cell of $L_\beta$ has diameter at most $\beta\sqrt{d}$.
Moreover any two non-adjacent cells have a distance at least 
$\beta\frac{\sqrt{2}}{d+1}$. 

We call a permutahedron \emph{full}, if it contains a point of $P$, and \emph{empty} otherwise
(we assume for simplicity that each point in $P$ lies in the interior of some permutahedron;
this can be ensured with well-known methods~\cite{sos}).
Clearly, there are at most $n$ full permutahedra for a given $P$.
We define $X_\beta$ as the nerve of the full permutahedra defined by $L_\beta$. 
An equivalent formulation is that $X_\beta$ is the subcomplex of $\Del$
defined in Section~\ref{section:permuto} induced by the lattice points
of full permutahedra.
This implies that $X_\beta$ is also a flag complex.
We usually identify the permutahedron and its center in $L_\beta$ 
and interpret the vertices of $X_\beta$ as a subset of $L_\beta$.
See Figure~\ref{fig:example} for an example in 2D.

\begin{figure}[h]
 \centering
 \includegraphics[scale=0.08]{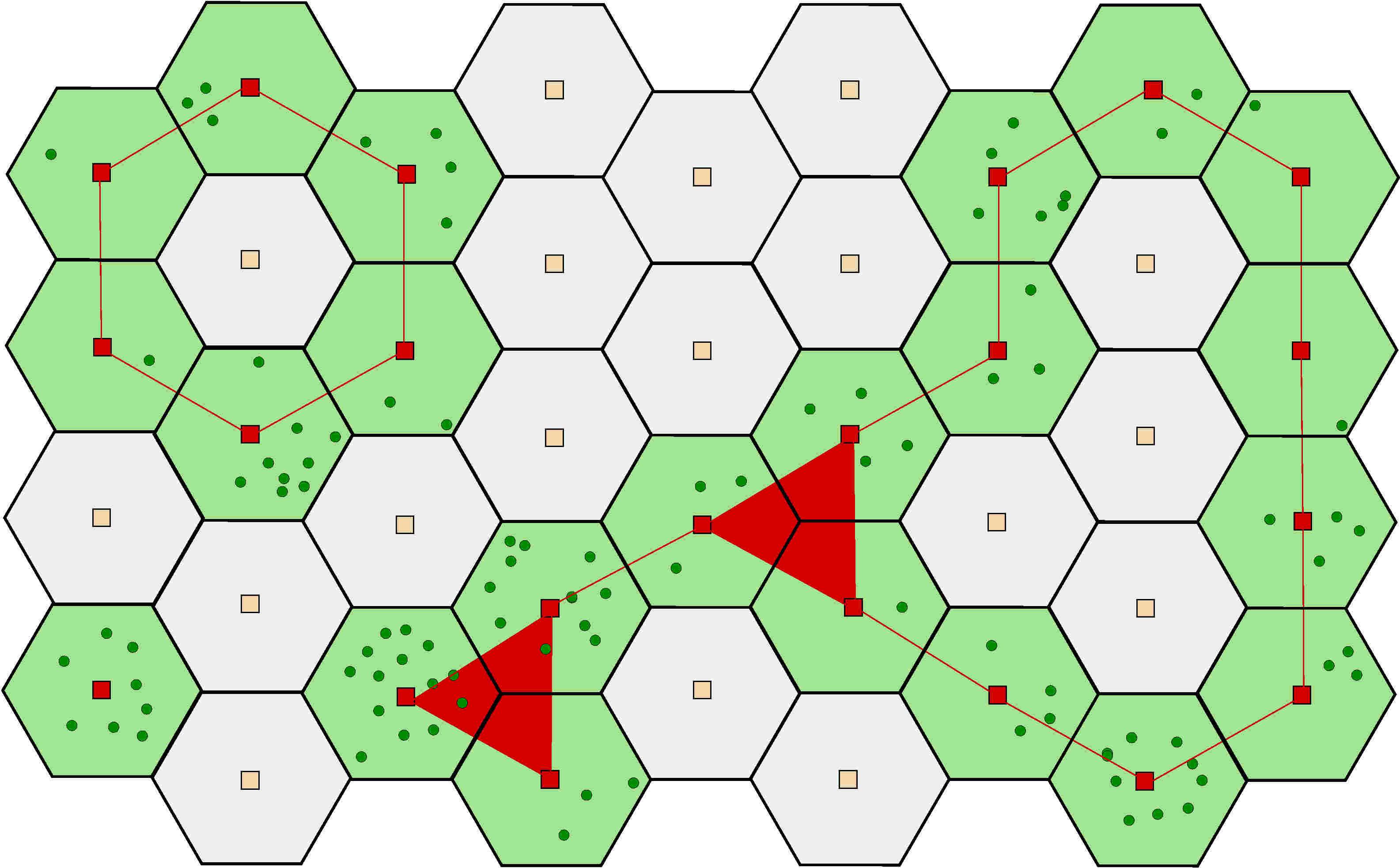}
 \caption{An example of $X_\beta$: the darkly shaded hexagons are the \emph{full}
 permutahedra, which contain input points marked as dark disks.
 Each dark square corresponds to a full permutohedron and 
 represents a vertex of $X_\beta$. If two full permutahedra
 are adjacent, there is an edge between the corresponding vertices. 
 The clique completion on the edge graph constitutes the complex $X_\beta$.}
 \label{fig:example}
\end{figure}

\subparagraph*{Interleaving}
To prove that $X_\beta$ approximates the Rips filtration, we define
simplicial maps connecting the complexes on related scales.

Let $V_{\beta}$ denote the subset of $L_\beta$ 
corresponding to full permutohedra.
To construct $X_\beta$, we use a map $v_\beta: P\rightarrow V_\beta$,
which maps each point $p\in P$ to its closest lattice point.
Vice versa, we define $w_\beta: V_\beta\rightarrow P$ to map a vertex in $V_\beta$
to the closest point of $P$. Note that $v_\beta\circ w_\beta$ is the identity map,
while $w_\beta\circ v_\beta$ is not.

\begin{lemma}
\label{lemma:appmaps_1}
The map $v_\beta$ induces a simplicial map $\phi_\beta:\ri_{\frac{\beta}{\sqrt{2}(d+1)}}
\rightarrow X_{\beta}$.
\end{lemma}
\begin{proof}
Because $X_\beta$ is a flag complex,
it is enough to show that for any edge $(p,q)$ in $\ri_{\frac{\beta}{\sqrt{2}(d+1)}}$,
$(v_\beta(p),v_\beta(q))$ is an edge of $X_{\beta}$.
This follows at once from the contrapositive of Lemma~\ref{lem:pmdist}.
\end{proof}

\begin{lemma}
\label{lemma:appmaps_2}
The map $w_\beta$ induces a simplicial map $\psi_\beta:X_{\beta} \rightarrow
\ri_{\beta 2\sqrt{d}}$.
\end{lemma}
\begin{proof}
It is enough to show that for any edge $(p,q)$ in $X_\beta$,
$(w_\beta(p),w_\beta(q))$ is an edge of $\ri_{\beta 2\sqrt{d}}$.
Note that $w_\beta(p)$ lies in the permutahedron of $p$ and similarly,
$w_\beta(q)$ lies in the permutahedron of $q$, so their distance
is bounded by twice the diameter of the permutahedron.
The statement follows from Lemma~\ref{lemma:pmdiam}.
\end{proof}
Since $\beta 2\sqrt{d}< \beta 2 (d+1)$, we can compose the map $\psi_\beta$ from the previous lemma
with an inclusion map to a simplicial map $X_{\beta} \rightarrow \ri_{\beta 2(d+1)}$ which we denote by $\psi_\beta$ as well.
Composing the simplicial maps $\psi$ and $\phi$, we obtain simplicial maps
\[\theta_\beta: X_\beta\rightarrow X_{\beta(2(d+1))^2}\]
for any $\beta$, giving rise to a discrete filtration
\[\left(X_{\beta(2(d+1))^{2k}}\right)_{k\in\Z}.\]

The maps define the following diagram of complexes and simplicial maps between them (we omit
the indices in the maps for readability):

 \begin{equation}
 \label{eqn:intl-orig}
  \xymatrix{
 & \cdots \ar[r] & \ri_{\beta2(d+1)} \ar[rd]^{\phi} \ar[rr]^{g} & &  \ri_{\beta8(d+1)^3} \ar[r] & \cdots \\
 \cdots \ar[r] & X_{\beta} \ar[ru]^{\psi} \ar[rr]^{\theta} & & X_{\beta4(d+1)^2} \ar[ru]^{\psi} \ar[r] & \cdots
 }
 \end{equation}

Here, $g$ is the inclusion map of the corresponding Rips complexes.
Applying the homology functor yields a sequence of vector spaces and linear maps between them.

\begin{lemma}
\label{lem:commutes}
Diagram \ref{eqn:intl-orig} commutes on the homology level, that is, $\theta_\ast=\phi_\ast\circ \psi_\ast$
and $g_\ast=\psi_\ast\circ\phi_\ast$, where the asterisk denotes the homology map induced by the simplicial map.
\end{lemma}
\begin{proof}
For the first statement, note that $\theta$ is defined as $\phi\circ\psi$, so the maps commute already
at the simplicial level.
The second identity is not true on a simplicial level; we show that the maps $g$ and $h:=\psi\circ\phi$ are \emph{contiguous},
that means, for every simplex $(x_0,\ldots,x_k)\in \ri_{\beta2(d+1)}$, the simplex $(g(x_0),\ldots,g(x_k),h(x_0),\ldots,h(x_k))$
forms a simplex in $\ri_{\beta8(d+1)^3}$. Contiguity implies that the induced homology maps $g_\ast$ and $h_\ast=\psi_\ast\circ\phi_\ast$
are equal~\cite[\S 12]{munkres}.

It suffices to prove that any pair of vertices among $\{g(x_0),\ldots,g(x_k),h(x_0),\ldots,h(x_k)\}$
is at most $\beta16(d+1)^3$ apart. This is immediately clear for any pair $(g(x_i), g(x_j))$
and $(h(x_i),h(x_j))$, so we can restrict to pairs of the form $(g(x_i), h(x_j))$.
Note that $g(x_i)=x_i$ since $g$ is the inclusion map. Moreover, $h(x_j)=\psi(\phi(x_j))$,
and $\ell:=\phi(x_j)$ is the closest lattice point to $x_j$ in $X_{\beta4(d+1)^2}$.
Since $\psi(\ell)$ is the closest point in $P$ to $\ell$, it follows that $\|x_j-h(x_j)\|\leq 2\|x_j-\ell\|$.
With Lemma~\ref{lemma:pmdiam}, we know that 
$\|x_j-\ell\|\leq \beta4(d+1)^2\sqrt{d}$, 
which is the diameter of 
the permutahedron cell. Using triangle inequality, we obtain
\[
\|g(x_i)-h(x_j)\|\leq \|x_i-x_j\|+\|x_j-h(x_j)\|\leq
\beta4(d+1)+\beta8(d+1)^2\sqrt{d}
<\beta16(d+1)^3
\]
\end{proof}
\begin{theorem}
 \label{theorem:appratio}
 The persistence module $\left(H_{\ast}(X_{\beta(2(d+1))^{2k}})\right)_{k\in\Z}$ 
 is a $6(d+1)$-approximation of $(H_{\ast}(\ri_{\beta}))_{\beta\geq 0}$.
\end{theorem}
\begin{proof}
Lemma~\ref{lem:commutes} proves that on the logarithmic scale, the two filtrations 
are \emph{weakly $\eps$-interleaved} with $\eps=2(d+1)$, in the sense
of~\cite{ch-proximity}. Theorem~4.3 of~\cite{ch-proximity} asserts
that the bottleneck distance of the filtrations is at most $3\eps$.
\end{proof}

\subparagraph*{Complexity bounds}
We exploit the non-degenerate configuration of the permutahedral tessellation to prove that $X_\beta$ is not too large.
We let $X_\beta^{(k)}$ denote the $k$-skeleton of $X_\beta$.

\begin{theorem}
\label{theorem:appsize}
For any $\beta$, $X_\beta^{(k)}$ has at most $n2^{O( d\log k)}$ simplices.
\end{theorem}
\begin{proof}
We fix $k$ and a vertex $v$ of $V_\beta$. Recall that $v$ represents a permutahedron, 
which we also denote by $\Pi(v)$.
By definition, any $k$-simplex containing $v$ corresponds to an intersection of $(k+1)$ permutahedra,
involving $\Pi(v)$. By Proposition~\ref{prop:equiv}, such an intersection
corresponds to a $(d-k)$-face of $\Pi(v)$. Therefore, the number of $k$-simplices involving $v$ 
is bounded by the number of $(d-k)$-faces of the permutahedron, which is $2^{O( d\log k)}$
using Lemma~\ref{lem:perm_no_of_faces}. 
The bound follows because $X_\beta$ has at most $n$ vertices.
\end{proof}
\begin{theorem}
\label{thm:complexity_approx}
For any $\beta$, $X_\beta^{(k)}$ can be computed in $O(n2^{d}+k^2 2^{d}|X_\beta^{(k)}|)$ time. 
In particular, the construction takes $n2^{O( d\log k)}$
in the worst case.
\end{theorem}

\begin{proof}
To find the vertices of $X_\beta$, we find,
for each $p\in P$, the closest point to $p$ in the scaled lattice $L_\beta$. For that, 
we use the algorithm from~\cite[Chap.20]{sloway-book} which
first finds the closest point in the coarser lattice $A_d$ and then
inspects a neighborhood of that lattice point to find the closest
point in $L_\beta$. 
This algorithm inspects at most $O(d^2)$ lattice points, thus finding the vertex set
runs in $O(nd^2)$ time.

To find the edges of $X_\beta$, we fix a vertex $v\in V_\beta$ and 
inspect all the $2^d$ neighbors, checking for each neighbor whether it
is in $V_\beta$ or not. This can be done in time $O(n2^d)$ time.

Finally, to find the higher-dimensional simplices, we simply compute
the flag complex over the obtained graph (Lemma~\ref{lem:flag_prop}).
For every $v\in V_\beta$ and any $k$-simplex $\sigma\in X_\beta$ involving $v$, 
we search for co-facets of $\sigma$: 
for every neighbor $w$ not involved in $X_\beta$, we test whether $w\ast\sigma$
is a $(k+1)$-simplex of $X_\beta$. This test is combinatorial and costs $O(k^2)$ time.
Consequently, for every simplex encountered, we spend an overhead of $O(k^2 2^d)$.
\end{proof}

\subparagraph{Dimension reduction}
For large $d$, our approximation complex plays nicely together with dimension reduction techniques. 
We start with noting that interleavings satisfy the triangle inequality.
This result is folklore; 
see~\cite[Thm 3.3]{bs-categorization} for a proof
in a generalized context.

\begin{lemma}
\label{lem:interleaving_transitivity}
Let $(A_\beta)$, $(B_\beta)$, and $(C_\beta)$ be persistence modules.
If $(A_\beta)$ is a $t_1$-approximation of $(B_\beta)$ and $(B_\beta)$ is a $t_2$-approximation of $(C_\beta)$,
then $(A_\beta)$ is a $(t_1t_2)$-approximation of $(C_\beta)$.
\end{lemma}

The following statement is a simple application of interleaving distances from~\cite{ch-proximity}.
We provide a proof in the appendix.

\begin{lemma}
\label{lem:generic_reduction}
Let $f: P \rightarrow \R^m$ be an injective map such that 
\[\xi_1\distance{p}{q} \le\distance{f(p)}{f(q)}\le \xi_2\distance{p}{q}\]
for some constants $\xi_1\leq 1\leq\xi_2$.
Let $\rid_{\alpha}$ denote the Rips complex of the point set $f(P)$.
Then, the persistence module $(H_{\ast}(\rid_{\alpha}))_{\alpha\geq 0}$ is an $\frac{\xi_2}{\xi_1}$-approximation of  
$(H_{\ast}(\ri_{\alpha}))_{\alpha\ge 0}$.
\end{lemma}

As a first application, we show that we can shrink the approximation size from Theorem~\ref{theorem:appsize}
for the case $d\gg \log n$, only worsening the approximation quality by a constant factor.

\begin{theorem}
\label{thm:jl_dr}
Let $P$ be a set of $n$ points in $\R^d$. There exists a constant $c$ and a discrete filtration of the form
$\left(\bar{X}_{(c\log n)^{2k}}\right)_{k\in\Z}$
that is $(3c\log n)$-interleaved with the Rips filtration of $P$ and at each scale $\beta$,
$\bar{X}^{(k)}_\beta$ has only $n^{O(\log k)}$ simplices. 
Moreover, we can compute, with high success probability, a complex $\bar{X}^{(k)}_\beta$ with this property
in deterministic running time $O(dn\log n)+k^2 n^{O(1)}|X^{(k)}_\beta| = n^{O(\log k)}$.
\end{theorem}
\begin{proof}
The famous lemma of Johnson and Lindenstrauss~\cite{jl-lemma} asserts the existence of a map $f$
as in Lemma~\ref{lem:generic_reduction} for $m=\lambda\log n/\eps^2$ with some absolute constant $\lambda$
and $\xi_1=(1-\eps)$, $\xi_2=(1+\eps)$. Choosing $\eps=1/2$, we obtain that $m=O(\log n)$ and $\xi_2/\xi_1=3$.
With $\rid_{\alpha}$ the Rips complex of the Johnson-Lindenstrauss transform, we have therefore that
$(H_{\ast}(\rid_\alpha))_{\alpha\ge 0}$ is a $3$-approximation of 
$(H_{\ast}(\ri_\alpha))_{\alpha\ge 0}$.
Moreover, using the approximation scheme from this section,
we can define a filtration $(\bar{X}_\beta)_{\beta\ge 0}$ whose induced persistence module
$(H_{\ast}(X_{\beta}))_{\beta\ge 0}$ is a $6(m+1)$-approximation of 
$(H_{\ast}(\rid_\alpha))_{\alpha\ge 0}$, and its size at each scale 
is $n2^{O(\log n\log k)}=n^{O(\log k)}$.
The first half of the result follows using Lemma~\ref{lem:interleaving_transitivity}.

The Johnson-Lindenstrauss lemma further implies that an orthogonal projection to a randomly
chosen subspace of dimension $m$ will yield an $f$ as above, with high probability.
Our algorithm picks such a subspace, projects all points into this subspace
(this requires $O(dn\log n)$ time) and applies the approximation scheme for the projected point set.
The runtime bound follows from Theorem~\ref{thm:complexity_approx}.
\end{proof}

Note that for $k=\log n$, the approximation complex from the previous theorem is of size $n^{O(\log\log n)}$
and thus super-polynomial in $n$. Using a slightly more elaborated dimension reduction result
by Matou\v{s}ek~\cite{mt-embedding}, we can get a size bound polynomial in $n$,
at the price of an additional $\log n$-factor in the approximation quality. Let us first state Matou\v{s}ek result (whose proof follows a similar strategy
as for the Johnson-Lindenstrauss lemma):

\begin{theorem}
\label{thm:matousek_dr_generalized}
Let $P$ be an $n$-point set in $\R^d$. Then, a random orthogonal projection into $\R^k$ for 
$3\le k\le C\log n$ distorts pairwise distances in $P$ by at most $O(n^{2/k}\sqrt{\log n/k})$. 
The constants in the bound depend only on $C$.
\end{theorem}
By setting $k:=\frac{4\log n}{\log\log n}$ in Matou\v{s}ek's result, we see that this results
in a distortion of at most $O(\sqrt{\log n \log\log n})$.

\begin{theorem}
\label{thm:matousek_dr}
Let $P$ be a set of $n$ points in $\R^d$. There exists a constant $c$ and a discrete filtration of the form
$\left(\bar{X}_{\big(c\log n \big(\frac{\log n}{\log \log n}\big)^{1/2}\big)^{2k}}\right)_{k\in\Z}$
that is $3c\log n \big(\frac{\log n}{\log \log n}\big)^{1/2}$-interleaved
with the Rips filtration on $P$ and at each scale $\beta$,
$\bar{X}^{(k)}_\beta$ has at most $n^{O(1)}$ simplices.
Moreover, we can compute, with high success probability, a complex $\bar{X}^{(k)}_\beta$ with this property
in deterministic running time $n^{O(1)}$.
\end{theorem}
\begin{proof}
The proof follows the same pattern of Theorem~\ref{thm:jl_dr} with a few changes. We use Matou\v{s}ek's dimension
reduction result described in Theorem~\ref{thm:matousek_dr_generalized} with the projection dimension being
$m:=\frac{4\log n}{\log\log n}$. Hence, $\xi_2/\xi_1=O(\sqrt{\log n \log\log n})$ for the Rips construction.
The final approximation factor is $6(m+1)\xi_2/\xi_1$ which simplifies to
$O(\log n \big(\frac{\log n}{\log \log n}\big)^{1/2})$.
The size and runtime bounds follow by substituting the value of $m$ in the respective bounds.
\end{proof}

Finally, we consider the important generalization that $P$ is not given as an embedding in $\R^d$,
but as a point sample from a general metric space. 
We use the classical result by Bourgain~\cite{bg-metric} to
embed $P$ in Euclidean space with small distortion. In the language of Lemma~\ref{lem:generic_reduction},
Bourgain's result permits an embedding into $m=O(\log^2 n)$ dimensions with a distortion $\xi_2/\xi_1=O(\log n)$,
where the constants are independent of $n$.
Our strategy for approximating a general metric space consists of first embedding
it into $\R^{O(\log^2 n)}$, then reducing the dimension,
and finally applying our approximation scheme on the projected embedding.
The results are similar to Theorems~\ref{thm:jl_dr} and \ref{thm:matousek_dr}, except that the
approximation quality further worsens by a factor of $\log n$ due to Bourgain's embedding.
We only state the generalized version of Theorem~\ref{thm:matousek_dr},
omitting the corresponding generalization of Theorem~\ref{thm:jl_dr}. The proof is straight-forward with the same techniques
as before.

\begin{theorem}
\label{thm:bourgain_dr_generalized}
Let $P$ be a general metric space with $n$ points. There exists a constant $c$ and a discrete filtration of the form
$\left(\bar{X}_{\big(c\log^2 n(\frac{\log n}{\log \log n})^{1/2}\big)^{2k}}\right)_{k\in\Z}$
that is $3c\log^2 n(\frac{\log n}{\log \log n})^{1/2} $-interleaved with the Rips filtration
on $P$ and at each scale $\beta$, $\bar{X}^{(k)}_\beta$ has at most $n^{O(1)}$ simplices. 
Moreover, we can compute, with high success probability, a complex $\bar{X}^{(k)}_\beta$ with this property
in deterministic running time $n^{O(1)}$.
\end{theorem}

\section{A lower bound for approximation schemes}
\label{section:lowerbnd}
We describe a point configuration for which 
the \Cech filtration gives rise to a large number, 
say $N$, of features with ``large'' persistence,
relative to the scale on which the persistence appears.
Any $\eps$-approximation of the \Cech filtration, for $\eps$ small enough,
has to contain at least one interval per such feature in its
persistent barcode, yielding a barcode of size at least $N$.
This constitutes a lower bound on the size of the approximation itself,
at least if the approximation stems from a simplicial filtration:
in this case, the introduction of a new interval in the barcode requires
at least one simplex to be added to the filtration; also more generally,
it makes sense to assume that any representation of a persistence module
is at least as large as the size of the resulting persistence barcode.

To formalize what we mean by a ``large'' persistent feature, 
we call an interval $(\alpha,\alpha')$ 
of $(H_\ast(\cech_\alpha))_{\alpha\geq 0}$ 
\emph{$\delta$-significant} for $0<\delta<\frac{\alpha'-\alpha}{2\alpha'}$.
Our approach from above translates into the following statement:
\begin{lemma}
\label{lem:significance_lemma}
For $\delta>0$, and a point set $P$, let 
$N$ denote the number of $\delta$-significant intervals of $(H_\ast(\cech_\alpha))_{\alpha\geq 0}$.
Then, any persistence module $(X_\alpha)_{\alpha\geq 0}$ 
that is an $(1+\delta)$-approximation of
$(H_\ast(\cech_\alpha))_{\alpha\geq 0}$ has at least $N$ intervals in its barcode.
\end{lemma}
\begin{proof}
If $(\alpha,\alpha')$ is $\delta$-significant, that means that 
there exist some $\eps>0$ and $c\in(\alpha,\alpha')$ such that
$\alpha/(1-\eps)\leq c/(1+\delta)<c(1+\delta)\leq\alpha'$.
Any persistence module 
that is an $(1+\delta)$-approximation of $(H_\ast(\cech_\alpha))_{\alpha\geq 0}$
needs to represent an approximation 
of the interval in the range $(c(1-\eps)/2,c)$; in other words, there is
an interval corresponding to $(\alpha,\alpha')$ in the approximation.
See the appendix for more details.
\end{proof}
\subparagraph*{Setup}
We next define our point set for a fixed dimension $d$.
Consider the $A^\ast$ lattice with origin $o$.
Recall that $o$ has $2^{d+1}-2$ neighbors in the Delaunay triangulation $\Del$
of $A_d^\ast$, because its dual Voronoi polytope, the permutahedron $\Pi$,
has that many facets. We define $P$ as the union of $o$ with all its Delaunay
neighbors, yielding a point set of cardinality $2^{d+1}-1$.
As usual, we set $n:=|P|$, so that $d=\Theta(\log n)$.

We write $\Del_P$ for the Delaunay triangulation of $P$.
Since $P$ contains $o$ and all its neighbors, the Delaunay simplices of 
$\Del_P$ incident to $o$ are the same as the Delaunay simplices of $\Del$
incident to $o$. Thus, according to Proposition~\ref{prop:equiv},
a $(k-1)$-simplex of $\Del_P$ incident to $o$ corresponds to
a $(d-k+1)$-face of $\Pi$ and thus to an ordered $k$-partition of $[d+1]$.

Fix a integer parameter $\ell\geq 3$, to be defined later.
We call an ordered $k$-partition $(S_1,\ldots,S_k)$ \emph{good}, 
if $|S_i|\geq \ell$ for every $i=1,\ldots,k$. 
We define good Delaunay simplices and good permutahedron faces
accordingly using Proposition~\ref{prop:equiv}.

Our proof has two main ingredients:
First, we show that a good Delaunay simplex either gives birth to or kills
an interval in the \Cech module that has a lifetime
of at least $\frac{\ell}{8(d+1)^2}$. This justifies our notion of ``good'', 
since good $k$-simplices create features that have to be preserved
by a sufficiently precise approximation.
Second, we show that there are $2^{\Omega(d\log \ell)}$ good $k$-partitions,
so good faces are abundant in the permutahedron.

\subparagraph*{Persistence of good simplices.}
Let us consider our first statement. 
Recall that $\alpha_\sigma$ is the filtration value 
of $\sigma$ in the \Cech filtration.
It will be convenient to have an upper bound for $\alpha_\sigma$.
Clearly, such a value is given by the diameter of $P$.
It is not hard to see the following bound (compare Lemma~\ref{lemma:pmdiam}), 
which we state for reference:
\begin{lemma}
\label{lem:P_diam}
The diameter of $P$ is at most $2\sqrt{d}$. 
Consequently, $\alpha_\sigma\leq 2\sqrt{d}$ for each simplex $\sigma$
of the \Cech filtration.
\end{lemma}

Recall that by fixing a simplex-wise filtration
of the \Cech filtration, it makes sense to talk about the persistence
of an interval associated to a simplex.
Fix a $(k-1)$-simplex $\sigma$ of $\Del_P$ incident to $o$ 
(which also belongs to the \Cech filtration).

\begin{lemma}
\label{lem:barycenter_lemma}
Let $f_\sigma$ be the $(d-k)$ face of $\Pi$ dual to $\sigma$,
and let $o_\sigma$ denote its barycenter.
Then, $\alpha_\sigma$ is the distance of $o_\sigma$ from $o$.
\end{lemma}
\begin{proof}
$o_\sigma$ is the closest point to $o$ on $f_\sigma$ because
$\vec oo_\sigma$
is orthogonal to $\vec po_\sigma  $ for any boundary vertex $p$ of $f_\alpha$.
Since $f_\sigma$ is dual to $\sigma$, all vertices of $\sigma$
are in same distance to $o_\sigma$.
\end{proof}
Recall $L_\sigma$ and $L^\ast_\sigma$
from Section~\ref{section:backgnd}
as the difference of the alpha value of $\sigma$ and its (co-)facets.

\begin{theorem}
\label{thm:pers_bound_for_good}
For a good simplex $\sigma$ of $\Del_P$, both $L_\sigma$ and $L^\ast_\sigma$
are at least $\frac{\ell}{24(d+1)^{3/2}}$.
\end{theorem}
\begin{proof}
We start with $L^\ast_\sigma$.
Let $\sigma$ be a $(k-1)$-simplex and let $S_1,\ldots,S_k$ be the
corresponding partition.
We obtain a co-facet $\tau$ of $\sigma$
through splitting one $S_i$ into two non-empty parts.

The main step is to bound the quantity $\alpha_\tau^2-\alpha_\sigma^2$.
By Lemma~\ref{lem:barycenter_lemma}, the alpha values are the squared
norms of the barycenters $o_\tau$ of $\tau$ and $o_\sigma$ of $\sigma$, 
respectively.
It is possible to derive an explicit expression 
of the coordinates of $o_\sigma$ and $o_\tau$.
It turns out that almost all coordinates are equal, and thus cancel out
in the sum, except at those indices that lie in the split set $S_i$.
Carrying out the calculations (as we do in the appendix), we obtain the bound
\[\alpha^2_\tau-\alpha^2_\sigma\geq\frac{(\ell-1)}{4(d+1)}.\]
 Moreover,
$\alpha_\tau\leq 2\sqrt{d}$ by Lemma~\ref{lem:P_diam}. This yields
\[\alpha_\tau-\alpha_\sigma = \frac{\alpha^2_\tau-\alpha^2_\sigma}{\alpha_\tau+\alpha_\sigma}\geq \frac{\alpha^2_\tau-\alpha^2_\sigma}{2\alpha_\tau}\geq \frac{\ell-1}{16(d+1)\sqrt{d}} \ge \frac{\ell}{24(d+1)^{3/2}}\]
for $\ell\geq 3$. The bound on $L_\sigma^\ast$ follows. 
For $L_\sigma$, note that
$\min_{\tau \text{ facet of }\sigma} L_\tau^\ast \leq L_\sigma,$
so it is enough to bound $L_\tau^\ast$ for all facets of $\sigma$. With $\sigma$ being a $(k-1)$-simplex, all but one of its facets
are obtained by merging two consecutive $S_i$ and $S_{i+1}$. 
However, the obtained partition is again good (because $\sigma$ is good),
so the first part of the proof yields the lower bound for all these facets.
It remains to argue about the facet of $\sigma$ that is not attached to
the origin. For this, we change the origin to any vertex of $\sigma$.
It can be observed 
(through the combinatorial properties of $\Pi$) that with respect to the new
origin, $\sigma$ has the representation $(S_j,\ldots,S_k,S_1,\ldots,S_{j-1})$,
thus the partition is cyclically shifted. In particular, $\sigma$
is still good with respect to the new origin. We obtain the missing facet
by merging the (now consecutive) sets $S_k$ and $S_1$, which is also a good
face, and the first part of the statement implies the result.
\end{proof}

As a consequence of Theorem~\ref{thm:pers_bound_for_good}, the interval
associated with a good simplex has length at least $\frac{\ell}{24(d+1)^{3/2}}$
using Lemma~\ref{lem:L_bound} and~\ref{lem:L_star_bound}.
Moreover, the interval cannot persist beyond the scale $2\sqrt{d}$
by Lemma~\ref{lem:P_diam}. It follows
\begin{corollary}
\label{cor:signi_good}
The interval associated to a good simplex is $\delta$-significant
for $\delta<\frac{\ell}{96(d+1)^2}$.
\end{corollary}
\subparagraph*{The number of good simplices.}
We assume for simplicity that $d+1$ is divisible by $\ell$. 
We call a good partition $(S_1,\ldots,S_k)$ \emph{uniform}, 
if each set consists of \emph{exactly} $\ell$ elements.
This implies that $k=(d+1)/\ell$.

\begin{lemma}
\label{lem:uniform_bound}
The number of uniform good partitions 
is exactly $\frac{(d+1)!}{\ell!^{(d+1)/\ell}}$.
\end{lemma}
\begin{proof}
Choose an arbitrary permutation and place the first $\ell$ entries in the $S_1$, the second $\ell$ entries
in $S_2$, and so forth. In each $S_i$, we can interchange the elements and obtain the same $k$-simplex.
Thus, we have to divide out $\ell!$ choices for each of the $(d+1)/\ell$ bins.
\end{proof}

We use this result to bound the number of good $k$-simplices in the following
theorem. To obtain the bound, we use estimates for the factorials 
using Stirling's approximation.
Moreover, we fix some constant $\rho\in (0,1)$ and set $\ell=(d+1)^\rho$.
After some calculations (see appendix), we obtain:

\begin{theorem}
\label{theorem:good_bound}
For any constant $\rho\in(0,1)$, $\ell=(d+1)^\rho$, $k=(d+1)/\ell$ and $d$ large enough, there exists a constant $\lambda\in (0,1)$ that only 
depends only on $\rho$, 
such that the number of good $k$-simplices is at least $(d+1)^{\lambda (d+1)}=2^{\Omega(d\log d)}$.
\end{theorem}
Putting everything together, we prove our lower bound theorem:
\begin{theorem}
\label{thm:lower_bound}
There exists a point set of $n$ points in $d=\Theta(\log n)$ dimensions,
such that any $(1+\delta)$-approximation of its \Cech filtration
contains $2^{\Omega(d\log d)}$ intervals in its persistent barcode,
provided that $\delta<\frac{1}{96(d+1)^{1+\eps}}$ with an arbitrary constant 
$\eps\in(0,1)$.
\end{theorem}
\begin{proof}
Setting $\rho:=1-\eps$, Theorem~\ref{theorem:good_bound} guarantees
the existence of $2^{\Omega(d\log d)}$ good simplices, all in a fixed dimension
$k$. In particular, the intervals of the \Cech persistence module associated
to these intervals are all distinct. Since $\ell=(d+1)^{1-\eps}$, 
Corollary~\ref{cor:signi_good} states that all these intervals are significant
because $\delta<\frac{1}{96d^{1+\eps}}=\frac{\ell}{96(d+1)^2}$.
Therefore, by Lemma~\ref{lem:significance_lemma}, any 
$(1+\delta)$-approximation of the \Cech filtration has $2^{\Omega(d\log d)}$
intervals in its barcode.
\end{proof}

Replacing $d$ by $\log n$ in the bounds of theorem, 
we see the number of intervals
appearing in any approximation super-polynomial is $n$
if $\delta$ is small enough.

\section{Conclusion}
\label{section:concl}
We presented upper and lower bound results on approximating Rips and \Cech
filtrations of point sets in arbitrarily high dimensions.
For \Cech complexes, the major result can be summarized as: for a dimension-independent
bound on the complex size,
there is no way to avoid a super-polynomial complexity for 
fine approximations of about $O(\log^{-1} n)$, while polynomial
size can be achieved for rough approximation of about $O(\log^2 n)$.

Filling in the large gap between the two approximation factors 
is an attractive avenue for future work.
A possible approach is to look at other lattices. It seems that 
lattices with good covering properties are 
correlated with a good approximation quality, and it may be worthwhile
to study lattices in higher dimension which improve largely on the covering
density of $A^\ast$ (e.g., the Leech lattice~\cite{sloway-book}).

Our approach, like all other known approaches, approximate also the 
geometry of the point set as a by-product, and we have to allow for
large error rates to overcome the curse of dimensionality.
An alternative approach to bridge the gap between upper and lower bounds
with an approximation scheme that only approximates topological features.

An unpleasant property of our approach is the dependence on the spread
of the point set. We pose the question whether it is possible to
eliminate this dependence by a more elaborate construction that avoids
the mere gluing of approximation complexes of consecutive scales.

\bibliographystyle{plain}

\newpage
\begin{appendix}

\section{Missing proofs}

\subparagraph*{Proof of Lemma~\ref{lemma:pmgenp}}
We rephrase the proof idea of~\cite{stanford-tech} in slightly simplified terms.
The representative vectors of $A_d^*$ are of the form
\[g_t=\frac{1}{(d+1)}(\underbrace{t,\cdots,t}_{d+1-t},\underbrace{t-(d+1),\cdots,t-(d+1)}_{t})\]
for $1\le t\le d$~\cite{sloway-book}.
It can be seen that each component of the numerator of $g_t$ is congruent to $t$ modulo $(d+1)$. 
Hence, we call the numerator of $g_t$ a remainder-$t$ point.
Since any lattice point $x$ in $A_d^*$ can be written as $x=\sum m_t\cdot g_t$, it follows that the numerator of $x$ is a 
remainder-\{$(\sum m_t\cdot t)$ modulo $(d+1)$\} point. 
  
Now, we show that the Delaunay cells of the $A^*_d$ lattice are all $d$-simplices, which will prove our claim.
Let $\vec{v}$ be a vertex of the permutahedron which is the Voronoi cell of the origin.  W.l.o.g, we can assume that
$\vec{v}=\frac{1}{2(d+1)}(d,d-2,\cdots,-d)$. The $A^*_d$ lattice points closest to $\vec{v}$ define the
Delaunay cell of $\vec{v}$.
We have seen that the lattice points have the form $\vec y=\frac{1}{d+1}\big(\vec{m}(d+1)+k\vec 1\big)$,
where $m\in \mathbb{Z}^{d+1}$. Also, $\dotp{m}{1}=-k$ as $\dotp{y}{1}=0$.

We wish to minimize the distance between $v$ and $y$ by choosing a suitable value for $m$. In other words, 
we wish to find $\text{argmin}_{\vec{m}}||\vec{y}-\vec{v}||^2$. Note that
\begin{eqnarray*}
\text{argmin}_{\vec{m}}||\vec{y}-\vec{v}||^2 &=& \text{argmin}_{\vec{m}} \sum (m_i+\frac{k}{d+1}-v_i)^2\\
&=& \text{argmin}_{\vec{m}} \sum (m_i-v_i)^2 + 2(m_i-v_i)\frac{k}{d+1}\\
&=& \text{argmin}_{\vec{m}} \sum (m_i - v_i)^2 + \frac{2k}{d+1}\sum m_i\\
&=& \text{argmin}_{\vec{m}} \sum (m_i - v_i)^2 + \frac{2k}{d+1}\cdot(-k)\\
&=& \text{argmin}_{\vec{m}} \sum (m_i - v_i)^2 \\
&=& \text{argmin}_{\vec{m}} ||\vec{m}-\vec{v}||^2 \\
&=& \text{argmin}_{\vec{m}} ||\vec{m}-\frac{1}{2(d+1)}(d,\cdots,-d)||^2 
\end{eqnarray*}

It can be verified that only lattice points with $\vec m \in \{0,-1\}^{d+1}$ are Delaunay neighbors
of the origin. Then, the lattice points closest to $\vec{v}$ are a subset of the Delaunay neighbors of the origin.
An elementary calculation shows that $||\vec{y}-\vec{v}||^2$ is minimized when
\[   
  \vec{m}=(\underbrace{0,\cdots,0}_{d+1-k},\underbrace{-1,\cdots,-1}_{k})
\]
for $\dotp{m}{1}=-k$.

This shows that there is a unique remainder-$k$ nearest lattice point to $\vec{v}$, for $k\in(0,\ldots,d)$. 
Also, it can be verified that each such lattice point is equidistant from $v$.
Hence, the Delaunay cell contains precisely $(d+1)$ points, one for each choice of $k$. 
The corresponding lattice points for any permutation $\pi$ of $v$ are also permutations, by the above derivation.
Hence, all such $d$-simplices are congruent.
This proves the claim.

\subparagraph*{Proof of Lemma~\ref{lem:hplane_facet}} 
Consider the $d$-simplex $\sigma$ incident to the
origin that is dual Voronoi vertex of $\Pi$ with coordinates
\[v=\frac{1}{d+1}\Big(d/2,d/2-1,\ldots,d/2-(d-1),d/2-d\Big).\]
The $(d-1)$-facet $\tau$ of $\sigma$ opposite to the origin is spanned by lattice points of the form
\[\ell_k=\frac{1}{(d+1)}(\underbrace{k,\cdots,k}_{d+1-k},\underbrace{k-(d+1),\cdots,k-(d+1)}_{k}), 1\le k \le d,\]
(see the proof of Lemma~\ref{lemma:pmgenp} above). 
All points in $V$ can be obtained by permuting the coordinates
of $\ell_k$.

We can verify at once that all these points lie on the hyperplane
$-x_1+x_{d+1}+1=0$, so this plane supports $\tau$.
The origin lies on the positive side of the plane.
All points in $V$ either lie on the plane
or are on the positive side as well, as one can easily check.
For the vertices of $\Pi$, observe that the value $x_1-x_{d+1}$ is minimized
for the point $v$ above, for which $x_1-x_{d+1}+1=1/(d+1)$ is obtained.
It follows that $v$ as well as any vertex of $V$ is 
at least in distance $\frac{1}{\sqrt{2}(d+1)}$ from $H$ (the $\sqrt{2}$
comes from the length of the normal vector). This proves the claim for
the simplex dual to $v$.

Any other choice of $\sigma$ is dual to a permuted version of $v$.
Let $\pi$ denote the permutation on $v$ that yields the dual vertex.
The vertices of $\tau$ are obtained by applying the same permutation
on the points $\ell_k$ from above. Consequently, the plane equation
changes to $-x_{\pi(1)}+x_{\pi(d+1)}+1=0$. The same reasoning as above
applies, proving the statement in general.

\subparagraph*{Details of the proof of Lemma~\ref{lem:flag_prop}}
We start with the proof of the second claim. Assume that $k$
facets $f_1,\ldots,f_k$ of $\Pi$ are pairwise intersecting.
For any facet $f_i$, there is a partition $(S_i,[d+1]\setminus S_i)$
associated to it. By Lemma~\ref{lem:pi_adjacency}, we have that either
$S_i\subset S_j$ or $S_j\subset S_i$ for each $i\neq j$.
This means that the $S_i$ are totally ordered, that means, there
exists an ordering $\pi$ of $\{1,\ldots,k\}$ such that
$S_{\pi(1)}\subset S_{\pi(2)}\subset\ldots\subset S_{\pi(k)}$.
Now, the partition
\[\left(S_{\pi(1)},S_{\pi(2)}\setminus S_{\pi(1)},S_{\pi(3)}\setminus S_{\pi(2)},\ldots,S_{\pi(k)}\setminus S_{\pi(k-1)},[d+1]\setminus S_{\pi(k)}\right)\] 
is a common refinement of all partitions, which implies
that the corresponding face is incident to all $k$ facets.
This proves the claim.

\smallskip

Now, we prove the first claim.
Let $(S_1,[d+1]\setminus S_1)$, $(S_2, [d+1]\setminus S_2)$ 
be the partitions defining facets $f_1$ and $f_2$ respectively. 
Since $f_1$ and $f_2$ are disjoint, we have that $S_1\not\subset S_2$ and 
$S_2\not\subset S_1$ by Lemma~\ref{lem:pi_adjacency}. 
Let us define the sets $T_1 = S_1\setminus S_2 $, $T_2 = S_2\setminus S_1$, 
$T_3 = S_1 \cap S_2$ and $T_4 = [d+1]\setminus S_1 \cup S_2$.
Also, let $|T_1|=a$, $|T_2|=b$ and $|T_3|=c$ with $a,b,c\ge 1$.
Then, $|T_4|=d+1-(a+b+c)$, $|S_1|=k:=a+c$ and $|S_2|=p:=b+c$.

Let $\ell_1$, $\ell_2$ denote the lattice points at
the centers of the permutahedra $\Pi_1$, $\Pi_2$ that
are attached to $\Pi$ on the faces $f_1$ and $f_2$, respectively.
We can derive the coordinates of $\ell_1$ and $\ell_2$ easily:
an elementary calculation shows that barycenter of the face $f_1$ has coordinates 
$\frac{k-(d+1)}{2(d+1)}=\frac{k}{2(d+1)}-\frac{1}{2}$ at indices in $S_1$ and
$\frac{k}{2(d+1)}$ at the rest of the positions.
Similarly, the barycenter of $f_2$ has coordinates $\frac{p-(d+1)}{2(d+1)}=\frac{p}{2(d+1)}-\frac{1}{2}$ 
at indices in $S_2$ and $\frac{p}{2(d+1)}$ otherwise.
Since $\Pi$ is centered at the origin, the 
coordinates of $\ell_1$ and $\ell_2$ are obtained by multiplying these coordinates with $2$.
See Table~\ref{tbl:l10-lnm} for details.

Let $B$ denote the bisector hyperplane between $\ell_1$ and $\ell_2$.
We show that $B$ is a separating hyperplane
for $\Pi_1$ and $\Pi_2$ with no point of either on the hyperplane,
which proves the claim.
The vector $n=(n_1,\ldots,n_{d+1}):=\ell_2-\ell_1$ is a normal vector to $B$.
Then, we define $B$ by $n\cdot(x-m)=0$ with $m=(\ell_1+\ell_2)/2$ being
the midpoint of $\ell_1$ and $\ell_2$. See Table~\ref{tbl:l10-lnm} for a 
description of $n$ and $m$.
\begin{table}[ht]
\begin{center}
\begin{tabular}{|c|c|c|c|c|c|}
  \hline
  indices & $\ell_2$ & $\ell_1$ & $n=\ell_2-\ell_1$ & $m = (\ell_2 + \ell_1)/2$ & count\\ \hline
  $T_1$ & $\frac{p}{d+1}$ & $\frac{k}{d+1}-1$ & $\frac{(p-k)}{d+1}+1=\alpha + 1$ & $\frac{(p+k)}{2(d+1)}-\frac{1}{2}=\beta-1/2$ & $a$ \\ \hline
  $T_2$ & $\frac{p}{d+1}-1$ & $\frac{k}{d+1}$ & $\frac{(p-k)}{d+1}-1=\alpha - 1$ & $\frac{(p+k)}{2(d+1)}-\frac{1}{2}=\beta-1/2$ & $b$ \\ \hline
  $T_3$ & $\frac{p}{d+1}-1$ & $\frac{k}{d+1}-1$ & $\frac{p-k}{d+1}=\alpha$ & $\frac{(p+k)}{2(d+1)}-1=\beta-1$ & $c$\\ \hline
  $T_4$ & $\frac{p}{d+1}$ & $\frac{k}{d+1}$ & $\frac{p-k}{d+1}=\alpha$ & $\frac{p+k}{2(d+1)}=\beta$ & $d+1-a-b-c$\\
  \hline
\end{tabular}
\end{center}
\caption{$\ell_1,\ell_2,n,m$}
\label{tbl:l10-lnm}
\end{table}

Since permutahedra tile space by translation, the vertices of $\Pi_1$ are of the form 
$x_1=\ell_1+\pi$ where $\pi$ is any permutation of 
$y=\frac{1}{d+1}\big(\frac{d}{2},\frac{d}{2}-1,\ldots,\frac{-d}{2}\big)$.
Writing $B(x_1):=n\cdot(x_1-m)$ for the function whose sign determines
the halfspace of $x_1$ with respect to $B$, we can write
$B(x_1)=B(\ell_1+\pi)=n\cdot(\ell_1+\pi-m)=n\cdot\ell_1-n\cdot m + n\cdot\pi$.
Similarly, for any vertex $x_2=\ell_2+\pi$ of $\Pi_2$,
$B(x_2)=n\cdot\ell_2-n\cdot m + n\cdot\pi$. 
We show that $B(x_1)<0$ and $B(x_2)>0$ for all permutations $\pi$, which proves the claim.
First, we calculate $n\cdot\ell_1$, $n\cdot\ell_2$ and $n\cdot m$
using Table~\ref{tbl:l10-lnm}: 
\[
 n\cdot \ell_1 =(\alpha+1)\big(\frac{k}{d+1}-1\big)a + (\alpha-1)\frac{k}{d+1}b + \alpha\big(\frac{k}{d+1}-1\big)c +
 \alpha\big(\frac{k}{d+1}\big)\{d+1-(a+b+c)\}\\
\]
Upon simplification, this reduces to 
$n\cdot \ell_1 = -a + \frac{k}{d+1}(a-b)$. 
Similarly, one can calculate that 
$ n\cdot \ell_2 = b + \frac{p}{d+1}(a-b)$. Next,
\[
 n\cdot m = (\alpha+1)(\beta-1/2)a+(\alpha-1)(\beta-1/2)b+\alpha(\beta-1)c+\alpha\beta[d+1-(a+b+c)]
\]
This simplifies to
$n\cdot m = -(a-b)\frac{(d+1)-(p+k)}{2(d+1)}$.
Subtracting, we get 
\[
 n\cdot\ell_1-n\cdot m = -a + \frac{k}{d+1}(a-b) + (a-b)\frac{(d+1)-(p+k)}{2(d+1)} 
\]
which reduces to 
$n\cdot\ell_1-n\cdot m = -\frac{a+b}{2} + \frac{(b-a)^2}{2(d+1)}$.

Since $\ell_1 - m=-(\ell_2 - m)$, hence $n\cdot(\ell_2-m)=-n\cdot(\ell_1-m)$.
Also, 
\[
 n\cdot\ell_1-n\cdot m = -\frac{a+b}{2} + \frac{(b-a)^2}{2(d+1)}
 < -\frac{a+b}{2} + \frac{(b+a)^2}{2(d+1)} <  -\frac{a+b}{2}\big(1-\frac{a+b}{d+1}\big)<0.
\]
Hence, $n\cdot\ell_1-n\cdot m$ is negative and $n\cdot\ell_2-n\cdot m$ is positive.
Substituting these values in $B(x_1)$ and $B(x_2)$, we get 
\[
B(x_1) = -\frac{a+b}{2} + \frac{(b-a)^2}{2(d+1)} + n\cdot\pi,
B(x_2) = \frac{a+b}{2} - \frac{(b-a)^2}{2(d+1)} + n \cdot\pi
\]
We now calculate the maximum absolute value of $n\cdot\pi$ and
use it to show that $B(x_1)$ is always negative and $B(x_2)$ is always positive.

The dot product $n\cdot\pi$ is obtained by first multiplying each component $y_i$ of the vector
$y=\frac{1}{d+1}\big(\frac{d}{2},\frac{d}{2}-1,\ldots,\frac{-d}{2}\big)$
with a component of $n$, which has one of 3 values: $\alpha+1$ for indices in $T_1$, $\alpha$ for $T_3\cup T_4$,
$\alpha-1$ for $T_2$ (refer Table~\ref{tbl:l10-lnm}); the intermediate products are then added up.
The permutation of $y$ maximizing  $n\cdot\pi$ follows from a simple arithmetic fact, 
which can be proved by a simple induction on the dimension of the vector.

\begin{lemma}
For any natural number $N\geq 2$,
let $V=(v_1,\ldots,v_N)$ and $W=(w_1,\ldots,w_N)$ be two vectors in $\R^N$ with $v_1\leq\ldots \leq v_N$
and $w_1\leq \ldots\leq w_N$. Let $\pi$ be a permutation over $[N]$, and let $\pi(W)$ be the 
vector with the corresponding permuted coordinates of $W$. Then, $\max_{\pi} \{V\cdot \pi(W)\}= V\cdot W$.
\end{lemma}

Let us denote the sum of the $q$ smallest components of $y$ by $N_q$ and the sum of the 
$q$ largest components of $y$ by $M_q$. It is easy to verify that $M_q+N_q=0$,
$N_q=N_{d+1-q}$ and $M_q=M_{d+1-q}$. Then,

\begin{align*}
 max(|n\cdot\pi|)&=(\alpha+1)M_a + (\alpha-1)N_b + \alpha(N_{d+1-a}-N_b) 
 = \alpha(M_a + N_{d+1-a} )+ M_a -N_b \\
 &= 0-N_a-N_b 
 = -\Big[ \frac{ a\{a-(d+1)\} }{2(d+1)} + \frac{ b\{b-(d+1)\} }{2(d+1)} \Big] \\
 &= \frac{a+b}{2}-\frac{a^2+b^2}{2(d+1)} 
 <  \frac{a+b}{2}-\frac{(b-a)^2}{2(d+1)}
\end{align*}
The last inequality implies that 
\[B(x_1)=-\frac{a+b}{2} + \frac{(b-a)^2}{2(d+1)} + n\cdot\pi <0,\]
and similarly, $B(x_2)>0$. The claim follows.

\subparagraph*{Proof of Lemma~\ref{lem:generic_reduction}}
The map $f$ is a bijection between $P$ and $f(P)$. The properties of $f$ ensure that the vertex maps
$f^{-1}$ and $f$, composed with appropriate inclusion maps, induce simplicial maps
\[  \rid_{\frac{\alpha}{\xi_2/\xi_1}}
\overset{\phi}\hookrightarrow \ri_{\alpha}
\overset{\psi}\hookrightarrow \rid_{\alpha\xi_2/\xi_1}.
\]
It is straightforward to show that the following diagrams commute on a simplicial level,
\begin{eqnarray}
\label{eqn:diag-dimred-rips}
\xymatrix{
\ri_{\frac{\alpha}{\beta}} \ar[rrr]^{g}\ar[rd]^{\psi} & & & \ri_{\beta\alpha'}    &  & \ri_{\beta\alpha} \ar[r]^{g} & \ri_{\beta\alpha'} \\
& \rid_\alpha \ar[r]^{g} & \rid_{\alpha'} \ar[ru]^{\phi} &                 & \rid_\alpha \ar[r]^{g} \ar[ru]^{\phi} & \rid_{\alpha'} \ar[ru]^{\phi}\\ 
& \ri_\alpha \ar[r]^{g} & \ri_{\alpha'} \ar[rd]^{\psi} &                 & \ri_\alpha \ar[r]^{g} \ar[rd]^{\psi} & \ri_{\alpha'} \ar[rd]^{\psi} \\ 
\rid_{\frac{\alpha}{\beta}} \ar[rrr]^{g}\ar[ru]^{\phi} & & & \rid_{\beta\alpha'}    &  & \rid_{\beta\alpha} \ar[r]^{g} & \rid_{\beta\alpha'}\\
}
\end{eqnarray}

where $g$ is the inclusion map.
Hence, the strong interleaving result from~\cite{ch-proximity} implies that both persistence modules are 
$\frac{\xi_2}{\xi_1}$-approximations of each other.

\subparagraph*{Details of the proof of Lemma~\ref{lem:significance_lemma}}
We first argue that $\delta$-significance implies the existence
of $\eps>0$ and $c\in(\alpha,\alpha')$ such that
$\alpha/(1-\eps)\leq c/(1+\delta)<c(1+\delta)\leq\alpha'$:
We choose $c:=\alpha'/(1+\delta)$, so that the last inequality is satisfied.
For the first inequality, we note first that
$(1-2\delta)<\frac{1}{(1+\delta)^2}$ for all $\delta<1/2$.
By assumption,
$\alpha'-\alpha>2\alpha'\delta$, so $\alpha<\alpha'(1-2\delta)<\frac{\alpha'}{(1+\delta)^2}=\frac{c}{1+\delta}$. Since the inequality is strict, we can
choose some small $\eps>0$, such that $\alpha/(1-\eps)\leq \frac{c}{1+\delta}$.

By the definition of $(1+\delta)$-approximation, we have a commutative diagram
\begin{eqnarray}
\xymatrix{
H(\cech_{c(1-\eps)/(1+\delta)}) \ar[rrr]^{g}\ar[rd]^{\phi} & & &  H(\cech_{c(1+\delta)}) \\
& X_{c(1-\eps)/2} \ar[r]^{h} & X_{c} \ar[ru]^{\psi} 
}
\end{eqnarray}
Let $\gamma$ be the element in the upper-left vector space, corresponding
to the $\delta$-significant interval. By definition, $g(\gamma)\neq 0$.
It follows that $h(\phi(\gamma))\neq 0$ either, 
so there is a corresponding interval in the approximation.

\subparagraph*{Details of the proof of Theorem~\ref{thm:pers_bound_for_good}}
Recall that $\alpha^2_\sigma$ is the squared length 
of the barycenter $o_{\sigma}$, and an analogue statement holds for $o_\tau$.
Also, recall that $\tau$ is obtained from $\sigma$ by splitting one $S_i$
in the corresponding partition $(S_1,\ldots,S_k)$ of $\sigma$.
Assume wlog that $S_k$ is split into $S'_k$ and $S'_{k+1}$  
(splitting any other $S_i$ yields the same bound)
and that $S_k$ is of size exactly $\ell$ (a larger cardinality only leads
to a larger difference).

Let $s_i:=|S_i|$ and $p_i=\sum_{j=1}^{i-1} |s_j|$.
Recall that $\Pi$ is spanned by a permutations 
of a particular point in $\R^{d+1}$, defined in Section~\ref{section:permuto};
we order these coordinates values by size in increasing order.
Then, the indices in $S_i$ will contain the coordinate values of order
$p_i+1,\ldots,p_i+s_i$. Writing $a_i$ for their average, the symmetric
structure of $\Pi$ implies that 
$o_\sigma$ has value $a_i$ in each coordinate $j\in S_i$.
Doing the same construction for $\tau$, we observe that the coordinates
of $o_\sigma$ and $o_\tau$ coincide for every coordinate $j\in S_1,\ldots,S_{k-1}$;
the only differences appear for coordinate indices of $S_k$, that is,
the partition set that was split to obtain $\tau$ from $\sigma$.
Writing $a_k$, $a'_k$, $a'_{k+1}$ 
for the average values of $S_k$, $S'_k$, $S'_{k+1}$, respectively, 
and $t:=|S'_k|$, we get 
\[\alpha^2_\tau-\alpha^2_\sigma=\sum_{i=1}^t \left( (a'_{k})^2-a_k^2\right) + \sum_{i=t+1}^\ell \left( (a'_{k+1})^2-a_k^2\right) 
= t\left( (a'_{k})^2- a_k^2\right) + (\ell-t)\left( (a'_{k+1})^2- a_k^2\right)\]
To obtain $a_k$, $a'_k$, and $a'_{k+1}$, we only need to compute 
the average of the appropriate coordinate values. A simple calculation shows
that $a_k=\frac{(d+1)-\ell}{2(d+1)}$, $a'_k=\frac{(d+1)-i}{2(d+1)}$ and $a'_{k+1}=\frac{(d+1)-\ell-i}{2(d+1)}$. Plugging in these values yields
\[\alpha^2_\tau-\alpha^2_\sigma=\frac{(d+1+\ell)t(\ell-t)}{4(d+1)^2},\]
whose minimum is achieved for $t=1$ (and $t=\ell-1$). Therefore, 
\[\alpha^2_\tau-\alpha^2_\sigma\geq \frac{(d+1+\ell)(\ell-1)}{4(d+1)^2}\geq 
\frac{\ell-1}{4(d+1)},\]
as claimed.

\subparagraph*{Proof of Theorem~\ref{theorem:good_bound}}

Recall from Lemma~\ref{lem:uniform_bound} that the number of good simplices
is at least 
\[\frac{(d+1)!}{\ell!^{(d+1)/\ell}}.\]
Stirling's approximation~\cite{stirling-wf} states that
\[\sqrt{2\pi} n^{n+1/2}e^{-n+1/(12n+1)} < n! < \sqrt{2\pi} n^{n+1/2}e^{-n+1/(12n)}.\]
We rephrase the upper bound as
\[ \sqrt{2\pi} n^{n+1/2}e^{-n+1/(12n)} \leq \sqrt{2\pi} e^{1/(12n)} n^{n+1/2}e^{-n} \leq e n^{n+1/2}e^{-n}\]
for $n\geq 2$ and the lower bound simply as
\[\sqrt{2\pi} n^{n+1/2}e^{-n+1/(12n+1)} \geq n^n e^{-n}.\]
In this way, we can lower bound the number of good simplices as
\begin{align}
\frac{(d+1)!}{\ell!^{(d+1)/\ell}} &\geq  \frac{(d+1)^{(d+1)}e^{-(d+1)}}{(e \ell^{\ell+1/2}e^{-\ell})^{(d+1)/\ell}}\\
&\geq  \frac{(d+1)^{(d+1)}e^{-(d+1)}}{e^{(d+1)/\ell} \ell^{(d+1)+(d+1)/(2\ell)}e^{-(d+1)}}\\
&\geq  \exp\left( (d+1)\log (d+1) - \frac{(d+1)}{\ell} - (d+1)\log\ell (1+\frac{1}{2\ell}) \right).
\end{align}
Choose $\ell=(d+1)^\rho$ with some constant $\rho<1$. The above simplifies to
\[ \exp\left( (d+1)\log (d+1) - (d+1)^{1-\rho} -
\rho (d+1)\log (d+1) (1+\frac{1}{2(d+1)^\rho}) \right) \]
\[
= \exp\left( (d+1)\log (d+1) (1-\rho(1+\frac{1}{2(d+1)^\rho})) - (d+1)^{1-\rho} \right).\]
Now, pick some $\lambda\in [0,1]$ such that $\rho<1-2\lambda<1$. We have 
that $\rho(1+\frac{1}{2(d+1)^\rho})<1-2\lambda$ for $d$ large enough. Thus, for $d$ large enough, 
\[\exp\left( (d+1)\log (d+1) (1-\rho(1+\frac{1}{2(d+1)^\rho})) - (d+1)^{1-\rho}\right) \]
\[
\geq \exp\left( 2\lambda (d+1)\log (d+1) -(d+1)^{1-\rho} \right)\geq \exp\left(\lambda (d+1)\log (d+1)\right).\]

\end{appendix}

\end{document}